\documentclass[letterpaper]{scrartcl} 

\usepackage[amsthm]{pamas}

\usepackage{tikz,
			pgflibraryshapes}
			\usetikzlibrary{positioning} 
			\usetikzlibrary{patterns,arrows,shapes}

\usepackage{latexsym,
			optparams}
			
\usepackage{array,
			xspace,
			multirow,
			hhline,
			graphicx,
			xcolor,
			colortbl,
			tabularx,
			booktabs,
			fixltx2e}
\usepackage{ifthen}
\usepackage{enumerate}
\usepackage{comment}
\usepackage{ragged2e}
\usepackage{mathrsfs}
\usepackage{balance}
\usepackage{mathrsfs}
\usepackage{eucal}
\usepackage{nicefrac}

\usetikzlibrary{shapes}
\usetikzlibrary{positioning,arrows}

\DeclareSymbolFont{largesymbolsA}{U}{txexa}{m}{n}

\usepackage{optparams,ifthen,xspace}

\usetikzlibrary{arrows}
   	\newcommand{\Pref}[1][]{
		\ifthenelse{\equal{#1}{}}{\mathrel \succsim}{\mathop{\succsim_{#1}}}
	}                                      
	\newcommand{\sPref}[1][]{                  
		\ifthenelse{\equal{#1}{}}{\mathrel \succ}{\mathop{\succ_{#1}}}
	}                                          
	\newcommand{\Indiff}[1][]{                 
		\ifthenelse{\equal{#1}{}}{\mathrel \sim}{\mathop{\sim_{#1}}}
	}
	\newcommand{\prefset}[1][]{\ifthenelse{\equal{#1}{}}{\mathcal{\succsim}}{\mathcal{\succsim}_{#1}}}

	\newcommand{\wpref}[1][] {\ifthenelse{\equal{#1}{}}{\mathrel{R}}{\mathrel{R_{#1}}}}
	\newcommand{\spref}[1][] {\ifthenelse{\equal{#1}{}}{\mathrel{P}}{\mathrel{P_{#1}}}}
	\newcommand{\indiff}[1][]{\ifthenelse{\equal{#1}{}}{\mathrel{I}}{\mathrel{I_{#1}}}}

	\newcommand{\ml}[1][]{\ifthenelse{\equal{#1}{}}{\mathit{ML}}{\mathit{ML}(#1)}}
	\newcommand{\sml}[1][]{\ifthenelse{\equal{#1}{}}{\mathit{SML}}{\mathit{SML}(#1)}}
	\newcommand{\sd}[1][]{\ifthenelse{\equal{#1}{}}{\mathit{SD}}{\mathit{SD}(#1)}}
	\newcommand{\rsd}[1][]{\ifthenelse{\equal{#1}{}}{\mathit{RSD}}{\mathit{RSD}(#1)}}
	\newcommand{\st}[1][]{\ifthenelse{\equal{#1}{}}{\mathit{ST}}{\mathit{ST}(#1)}}
	\newcommand{\bd}[1][]{\ifthenelse{\equal{#1}{}}{\mathit{BD}}{\mathit{BD}(#1)}}
	\newcommand{\pc}[1][]{\ifthenelse{\equal{#1}{}}{\mathit{PC}}{\mathit{PC}(#1)}}
	\newcommand{\dl}[1][]{\ifthenelse{\equal{#1}{}}{\mathit{DL}}{\mathit{DL}(#1)}}
	\newcommand{\ul}[1][]{\ifthenelse{\equal{#1}{}}{\mathit{UL}}{\mathit{UL}(#1)}}	
	\newcommand{\uc}[1][]{\ifthenelse{\equal{#1}{}}{\mathit{UC}}{\mathit{UC}(#1)}}
	\newcommand{\bor}[1][]{\ifthenelse{\equal{#1}{}}{\mathit{BOR}}{\mathit{BOR}(#1)}}
	\newcommand{\serdict}[1][]{\ifthenelse{\equal{#1}{}}{\sigma}{\sigma(#1)}}
	
\renewcommand{\epsilon}{\varepsilon}

\begin{document}
	
\title{Welfare Maximization Entices Participation}

\author{Florian Brandl \qquad Felix Brandt \qquad Johannes Hofbauer\\
Technische Universit\"at M\"unchen\\
\texttt{\small \{brandlfl,brandtf,hofbauej\}@in.tum.de}
}

\date{}

\maketitle

\begin{abstract}
We consider randomized mechanisms with optional participation. Preferences over lotteries are modeled using skew-symmetric bilinear (SSB) utility functions, a generalization of classic von Neumann-Morgenstern utility functions.
We show that every welfare-maximizing mechanism entices participation and that the converse holds under additional assumptions. 
Two important corollaries of our results are characterizations of an attractive randomized voting rule that satisfies Condorcet-consistency and entices participation. This stands in contrast to a well-known result by \citet{Moul88b}, who proves that no deterministic voting rule can satisfy both properties simultaneously.
\end{abstract}

	\section{Introduction}\label{sec:introduction}

	Let $\mathbb{N}=\{1,2,\dots\}$ be a countable set of agents and $\mathcal{F}(\mathbb{N})$ the set of all finite and non-empty subsets of $\mathbb{N}$. Moreover, $A$ is a finite set of alternatives and $\Delta(A)$ the set of all \emph{lotteries} (or \emph{probability distributions}) over $A$. 
	A lottery is \emph{degenerate} if it puts all probability on a single alternative.
	We assume that preferences over lotteries are given by \emph{skew-symmetric bilinear (SSB) utility} functions as introduced by \citet{Fish82c}. 
	An SSB function $\phi$ is a function from $\Delta (A) \times \Delta (A) \rightarrow \mathbb{R}$ that is skew-symmetric and bilinear, 
	\ie
	\begin{align*}
	\phi(p,q) &= - \phi(q, p)\text,\\
	\phi(\lambda p + (1-\lambda) q, r) &= \lambda \phi(p,r) + (1-\lambda) \phi(q,r)\text.
	\end{align*}
	for all $p,q\in \Delta(A)$ and $\lambda\in\mathbb{R}$.
	Note that, by skew-symmetry, linearity in the first argument implies linearity in the second argument and that, due to bilinearity, $\phi$ is completely determined by its function values for degenerate lotteries.
	SSB utility theory is more general than the linear expected utility theory due to \citet{vNM47a}, henceforth vNM, as it does not require independence and transitivity \citep[see, \eg][]{Fish88a,Fish84c,Fish84e,Fish82c}.
	Hence, every vNM function $u$ is equivalent to an SSB function $\phi^{u}$, where $\phi^{u}(p,q) = u(p) - u(q)$, in the sense that both functions induce the same preferences over lotteries.
	In general, let $\Phi\subseteq \mathbb{R}^{A\times A}$ be a set of possible utility functions called the \emph{domain}.

	For every $N\in\mathcal{F}(\mathbb{N})$, let $\phi_N = (\phi_i)_{i\in N}\in \Phi^N$ be a vector of SSB functions. If $N = \{i\}$, we write $\phi_i$ with abuse of notation. 
	A lottery \emph{$p$ is welfare-maximizing for~$\phi_N$} if
	\[
	\sum_{i\in N} \phi_i(p,q)\ge 0 \text{ for all } q\in\Delta(A)\text.\tag{welfare maximization}
	\]
	If agents are endowed with vNM functions, welfare maximization is equivalent to maximizing the sum of expected utilities. In this case, there always exists a \emph{degenerate} welfare-maximizing lottery. When moving to general SSB functions, this does not hold anymore, but the existence of a (not necessarily degenerate) welfare-maximizing lottery is guaranteed by the minimax theorem \citep{Fish84a}.

	Our central objects of study are \emph{mechanisms} that map a vector of SSB functions to a lottery. 
	A mechanism is \emph{welfare-maximizing} if it always returns welfare-maximizing lotteries for vectors of SSB functions from $\Phi$, \ie for all $N\in\mathcal{F}(\mathbb{N})$ and $\phi_N\in\Phi^N$, $f(\phi_N)$ is welfare-maximizing.
	
	We will relate welfare maximization to participation. A mechanism \emph{satisfies participation} if participating in the mechanism never decreases the welfare of the participating group of agents. Formally, for every $N\in\mathcal{F}(\mathbb{N})$, $S\subsetneq N$, and $\phi_N\in\Phi^N$,
	\[
		\sum_{i\in S} \phi_i(f(\phi_N),f(\phi_{N\setminus S}))\ge 0\text.\tag{participation}
	\]

As we will see in \secref{sec:ordinal}, this strong notion of participation has important consequences even in settings in which the interpersonal comparison of utility is problematic (such as in voting).	

	\section{Welfare Maximization and Participation}
	\label{sec:main}

	We are now ready to prove three theorems that highlight the relationship between welfare maximization and participation. The first result shows that welfare maximization implies participation. While this is straightforward for vNM utility functions, the generalization to SSB functions will be vital for the results in \secref{sec:ordinal}.

	\begin{theorem}\label{thm:main}
		Every welfare-maximizing mechanism satisfies participation.
	\end{theorem}

	\begin{proof}		
		Let~$N\in\mathcal{F}(\mathbb{N})$,~$S\subsetneq N$,~$\phi_N\in\Phi^N$, and $f$ a welfare-maximizing mechanism.
		Furthermore, let
		\[
			\phi^N = \sum_{i\in N} \phi_i \quad \text{and} \quad \phi^S = \sum_{i\in S} \phi_i \quad \text{and} \quad \phi^{N\setminus S} = \sum_{i\in N\setminus S}\phi_i\text.
		\]
		For $p = f(\phi_N)$ and $p' = f(\phi_{N\setminus S})$, we then have that
		\begin{align*}
			\phi^N(p,q)&\ge 0 \text{ for all } q\in\Delta(A) \text{, and}\\
			\phi^{N\setminus S}(p',q)&\ge 0 \text{ for all } q\in\Delta(A)\text,
		\end{align*}
		since, by assumption, $f$ is welfare-maximizing for $\phi_N$ and $\phi_{N\setminus S}$.
		Thus, it follows that
		\begin{align*}
			\phi^S(p,p') = \phi^N(p,p') - \phi^{N\setminus S}(p,p') = \underbrace{\phi^N(p,p')}_{\ge 0} + \underbrace{\phi^{N\setminus S}(p',p)}_{\ge 0} \ge 0\text.
		\end{align*}
		The second equality follows from skew-symmetry of $\phi^{N\setminus S}$. The inequality follows from the fact that $f$ is welfare-maximizing for $\phi_N$ and $\phi_{N\setminus S}$.
		Hence, $f$ satisfies participation.
		\end{proof}

Clearly, \thmref{thm:main} also holds for Cartesian domains that are not symmetric among agents.
The converse of \thmref{thm:main} does not hold in full generality as every constant function satisfies participation but fails to be welfare-maximizing.
However, for sufficiently rich domains, the converse holds for mechanisms that satisfy additional properties. 

We first define homogeneity and weak welfare maximization.
	For all $N\in\mathcal{F}(\mathbb{N})$ and $k\in\mathbb{N}$, let
	$kN = \{i + l \max(N)\colon i\in N\text{ and } l\in[k]\}$
	and $\phi_i = \phi_j$ if $i\equiv j\mod \max(N)$.
	A mechanism $f$ is \emph{homogeneous} if replicating the set of agents does not affect the outcome, \ie $f(\phi_N) = f(\phi_{kN})$ for all $k\in\mathbb{N}$ and $\phi_N\in\Phi^N$.
	A mechanism $f$ is \emph{weakly welfare-maximizing} if $f(\phi_N)=p$ whenever $p$ is degenerate and the unique welfare-maximizing lottery for $\phi_N$. 
	The following lemma shows that a degenerate lottery is the unique welfare-maximizing lottery if and only if it is strictly preferred to every other degenerate lottery.

	\begin{lemma}\label{lem:condorcet}
		Let $\phi\in\Phi$ and $x\in A$. $x$ is the unique welfare-maximizing lottery for $\phi$ iff $\phi(x,y) > 0$ for all $y\in A\setminus\{x\}$. 
	\end{lemma}
	
	\begin{proof}
		For the direction from left to right, 
		assume that $x\in A$ is welfare-maximizing for $\phi$ and $B = \{y\in A\setminus\{x\}\colon \phi(x,y) \le 0\}\neq \emptyset$. By $p'$ we denote some lottery on $B$ that is welfare-maximizing for $(\phi_{ij})_{i,j\in B}$. Let $p$ be the lottery that is equal to $p'$ on $B$ and $0$ otherwise, \ie $p(y) = p'(y)$ for all $y\in B$ and $p(y) = 0$ for all $y\in A\setminus B$. By the choice of $B$, we have that $\phi(p, y)\ge 0$ for all $y\in B\cup\{x\}$. Moreover, for $\epsilon > 0$ small enough, we have $\phi(\epsilon p + (1-\epsilon) x, y)\ge 0$ for all $y\in A\setminus B$, since $\phi(x,y) > 0$. Hence, $\epsilon p + (1-\epsilon) x$ is also welfare-maximizing for $\phi$ and $x$ cannot be the unique welfare-maximizing lottery.
		
		The direction from right to left follows from linearity of $\phi$.
	\end{proof}

Next, we define two conditions on domains.
A domain $\Phi$ is \emph{symmetric} if for all $\phi\in\Phi$, $-\phi\in\Phi$. A domain $\Phi$ is \emph{non-imposing} if for all $x\in A$, there is $\phi\in\Phi$ such that $\phi(x,y)>0$ for all $y\in A$.

\begin{theorem}\label{thm:condorcet}
Let $\Phi$ be a symmetric and non-imposing domain.
Every homogeneous, weakly welfare-maximizing mechanism on $\Phi$ that satisfies participation is welfare-maximizing.
\end{theorem}

\begin{proof}
	Let $f$ be a homogeneous, weakly welfare-maximizing mechanism that satisfies participation and 
	assume for contradiction that $f$ is not welfare-maximizing for some $N\in\mathcal{F}(\mathbb{N})$ and $\phi_N\in\Phi^N$, \ie there is a lottery $q$ such that $\sum_{i\in N} \phi_i(p,q) < 0$, where $p = f(\phi_N)$.
	By linearity of the $\phi_i$, there is an alternative $x$ such that $\sum_{i\in N} \phi_i(p,x) = c < 0$. 
	Now, let $\bar N$ be a set of agents disjoint from $N$ such that $\phi_{\bar N} = -\phi_N$ and $j$ an agent that is not contained in $N$ or $\bar N$ such that $x$ is a degenerate unique welfare-maximizing lottery for $\phi_j$, \ie $\phi_j(x,y)>0$ for all $y\in A\setminus\{x\}$. 
	Such an agent exists since $\Phi$ is non-imposing.
	Moreover, let $d = \max\{\phi_j(x,y)\colon y\in A\}$ and $k$ be an integer such that $kc + d < 0$.
	It follows from homogeneity that $f(\phi_{kN}) = f(\phi_N)$. 
	By definition of $\phi_{\bar N}$ and $\phi_j$, it follows that $x$ is the unique welfare-maximizing lottery for $\phi_{kN\cup k\bar N\cup\{j\}}$. 
	Hence, $x = f(\phi_{kN\cup k\bar N\cup\{j\}})$ follows from weak welfare maximization of $f$. Furthermore, we have
	\[
		\sum_{i\in \bar N} k\phi_i(p,x) + \phi_j(p,x) \ge -(kc + d) > 0\text,
	\]
	which contradicts participation.
\end{proof}

For the second characterization, we define non-imposition of mechanisms and a property called cancellation, which requires that adding two agents with completely opposed preferences does not affect the outcome.
A mechanism $f$ is \emph{non-imposing} if for all $N\in\mathcal{F}(\mathbb{N})$ and all degenerate lotteries $p\in \Delta(A)$, there is $\phi_N\in \Phi^N$ such that $f(\phi_N)=p$. A mechanism $f$ satisfies \emph{cancellation} if for all $N\in\mathcal{F}(\mathbb{N})$, $\phi_N\in\Phi^N$, and $i,\bar i\not\in N$ such that $\phi_{i} = -\phi_{\bar i}\in\Phi$, $f(\phi_{N\cup\{i,\bar i\}})=f(\phi_N)$.
	
\begin{theorem}	\label{thm:cancellation}
Let $\Phi$ be a symmetric domain.
Every homogeneous, non-imposing mechanism that satisfies cancellation and participation is welfare-maximizing.
\end{theorem}
\begin{proof}
		Let $f$ be a homogeneous, non-imposing mechanism that satisfies cancellation and participation and 
	assume for contradiction that $f$ is not welfare-maximizing for some $N\in\mathcal{F}(\mathbb{N})$ and $\phi_N\in\Phi^N$, \ie there is a lottery $q$ such that $\sum_{i\in N} \phi_i(p,q) < 0$, where $p = f(\phi_N)$.
	By linearity of the $\phi_i$, there is an alternative $x$ such that $\sum_{i\in N} \phi_i(p,x) = c < 0$. 
	Now, let $\bar N$ be a set of agents disjoint from $N$ such that $\phi_{\bar N} = -\phi_N$. 
	Moreover, let $\{j\}$ be a set of agents disjoint from $N$ and $\bar N$ such that $f(\phi_j) = x$ (which exists by non-imposition) and let $d = \max\{\phi_j(x,y)\colon y\in A\}$. 
	Then, let $k$ be an integer such that $kc + d < 0$.
	It follows from homogeneity that $f(\phi_{kN}) = f(\phi_N)$. 
	Since $f$ satisfies cancellation, we have that $f(\phi_{kN\cup k\bar N\cup \{j\}}) = f(\phi_{N'}) = x$. 
	Furthermore, we have
	\[
		\sum_{i\in \bar N} k\phi_i(p,x) + \phi_j(p,x) \ge -(kc + d) > 0\text,
	\]
	which contradicts participation.
\end{proof}

	\section{Ordinal Mechanisms}
	\label{sec:ordinal}

We now turn to the important special case in which only ordinal preferences between alternatives are known and consider \emph{ordinal mechanisms}, \ie functions that map an ordinal preference profile to a lottery. 
Ordinal preferences are given in the form of complete, reflexive, and transitive binary relations, which can be conveniently represented by SSB functions whose entries are restricted to $\{-1,0,+1\}$, where $\phi_i(x,y)=+1$ if agent $i$ prefers $x$ to $y$, $\phi_i(x,y)=-1$ if he prefers $y$ to $x$, and $\phi_i(x,y)=0$ if he is indifferent. We refer to this representation as the \emph{canonical} utility representation of ordinal preferences and define ordinal mechanisms on the domain
\[\Phi_{\pc} = \{-1,0,+1\}^{A\times A}\text.\]
Every such representation entails a complete preference relation over lotteries of alternatives (called the \emph{pairwise comparison (PC)} preference extension). The natural interpretation of this relation is that lottery $p$ is preferred to lottery $q$ if the probability that $p$ yields an alternative preferred to the alternative returned by $q$ is at least as large as the other way round. For more details, please see \citet{Blav06a} and \citet{ABB13d,ABB14b}.

A lottery \emph{stochastically dominates} another if the former yields more expected utility than the latter for every vNM function that is consistent with the ordinal preferences. An ordinal mechanism satisfies \emph{ordinal participation} if no group of agents can abstain from $f$ such that \emph{each} of the agents is individually better off with respect to stochastic dominance \citep[see][]{BBH15b}.

\begin{proposition}\label{pro:ordinalpart}
Every ordinal mechanism that satisfies participation satisfies ordinal participation.
\end{proposition}
		\begin{proof}			
			Let~$N\in\mathcal{F}(\mathbb{N})$,~$S\subsetneq N$,~$\phi_N\in\Phi_{\pc}^N$, and $f$ an ordinal mechanism that satisfies participation.
			Participation of $f$ implies
			$
			\sum_{i\in N} \phi_i (f(\phi_N), f(\phi_{N\setminus S})) \ge 0\text.
			$
			In particular, there is $i\in S$ such that $\phi_i(f(\phi_{N}), f(\phi_{N\setminus S})) \ge 0$. \citet{ABB14b} have shown that the preference relation induced by the canonical utility representation is a refinement of stochastic dominance, \ie if $p$ stochastically dominates $q$, then $p$ is also preferred to $q$ with respect to the canonical utility representation. Hence, $i$ weakly prefers participating to abstaining with respect to stochastic dominance.
		\end{proof}

Ordinal participation is not easily satisfied. For example, \citet{BBH15b} have shown that no majoritarian ordinal mechanism can satisfy ordinal participation and \emph{ex post} efficiency.\footnote{An ordinal mechanism is majoritarian if its output only depends on the pairwise majority relation. \emph{Ex post} efficient mechanism always assign probability $0$ to Pareto dominated alternatives.}
By leveraging the results obtained in \secref{sec:main}, we can derive a number of statements concerning ordinal mechanisms that return so-called maximal lotteries. A lottery is \emph{maximal} for a given ordinal preference profile if it is welfare-maximizing for the canonical utility representation. In the context of voting, maximal lotteries are almost always unique and ordinal mechanisms that return maximal lotteries form an attractive class of randomized voting rules \citep{Fish84a,ABBH12a,Bran13a}. Such mechanisms have also been considered in the context of randomized assignment \citep{KMN11a,ABS13a}.

First, \thmref{thm:main} and \propref{pro:ordinalpart} imply that every ordinal mechanism that returns maximal lotteries satisfies ordinal participation. 

\begin{corollary}
Every ordinal mechanism that returns maximal lotteries satisfies ordinal participation.
\end{corollary}

Theorems \ref{thm:condorcet} and \ref{thm:cancellation} entail axiomatic characterizations of mechanisms that return maximal lotteries. 
It is easily seen that $\Phi_{\pc}$ satisfies symmetry and non-imposition. 
Alternative $x$ is called a \emph{Condorcet winner} of a given preference profile $\phi_N\in \Phi_{\pc}^N$ if a majority of agents prefers it to any other alternative, \ie $\sum_{i\in N} \phi_i(x,y)>0$ for all $y\in A$. An ordinal mechanism is \emph{Condorcet-consistent} if it always puts probability $1$ on a Condorcet winner. It follows from \lemref{lem:condorcet} that Condorcet-consistency is equivalent to weak welfare maximization for the domain $\Phi_{\pc}$. We thus obtain the following characterization as a corollary of \thmref{thm:condorcet}.
	
		\begin{corollary}\label{cor:condorcet}
			Every homogeneous, Condorcet-consistent, ordinal mechanism that satisfies participation returns maximal lotteries.
		\end{corollary}
		
\coref{cor:condorcet} can be contrasted with a classic result by \citet{Moul88b}. \citeauthor{Moul88b} has shown that no Condorcet-consistent ordinal mechanism that always returns degenerate lotteries satisfies ordinal participation.

Similarly, \thmref{thm:cancellation} yields an alternative characterization of maximal lotteries. Observe that non-imposition is weaker than Condorcet-consistency whereas cancellation is independent from Condorcet-consistency.
	
\begin{corollary}\label{cor:cancellation}
Every homogeneous, non-imposing, ordinal mechanism that satisfies cancellation and participation returns maximal lotteries.
\end{corollary}

We remark that Corollaries \ref{cor:condorcet} and \ref{cor:cancellation} do not hold if participation is weakened to ordinal participation. For example, the mechanism that returns the Condorcet winner if one exists and the uniform lottery over all alternatives otherwise is homogeneous, Condorcet-consistent, non-imposing and satisfies cancellation and ordinal participation. However, this mechanism violates \emph{ex post} efficiency let alone the stronger notion of ordinal efficiency, which is satisfied by maximal lotteries.

	\section*{Acknowledgments}
	This material is based upon work supported by Deutsche Forschungsgemeinschaft under grant {BR~2312/10-1}.
	
	\bibliographystyle{plainnat}
	% \bibliography{../pamas/abb,../pamas/brandt,../pamas/pamas}

		% 
		% 
		% 
		% 
		% 
		% 
		% 
		% 
		% 
		% 
		% 
		% 
		% 
		% 
		% 
		% 
		% 
		% 
		% 
		% 
		% 
		% 
		% 
		% 
		% 
		% 
		% 
		% 
		% 
		% 
		% 

\end{document}